\documentclass{article}
\usepackage{ijcai22}

\usepackage{times}
\usepackage{soul}
\usepackage{url}
\usepackage[T1]{fontenc}
\usepackage[hidelinks]{hyperref}
\usepackage[utf8]{inputenc}
\usepackage[small]{caption}
\usepackage{graphicx}
\usepackage{amsmath}
\usepackage{amsthm}
\usepackage{booktabs}
\usepackage{algorithm}
\usepackage{algorithmic}
\newtheorem{theorem}{Theorem}

\usepackage{amssymb,mathrsfs,mathtools}

\newtheorem{lemma}{Lemma}

\title{On Scheduling Mechanisms Beyond the Worst Case}

\author{
 Yansong Gao$^1$
\and
Jie Zhang$^2$
\affiliations
$^1$Applied Mathematics and Computational Science, University of Pennsylvania\\
$^2$Electronics and Computer Science, University of Southampton 
\emails
gaoyans@sas.upenn.edu,
jie.zhang@soton.ac.uk
}

\begin{document}

\maketitle

\begin{abstract}
The problem of scheduling unrelated machines has been studied since the inception of algorithmic mechanism design~\cite{NR99}. It is a resource allocation problem that entails assigning $m$ tasks to $n$ machines for execution. Machines are regarded as strategic agents who may lie about their execution costs so as to minimize their allocated workload. To address the situation when monetary payment is not an option to compensate the machines' costs, \citeauthor{DBLP:journals/mst/Koutsoupias14} [2014] devised two \textit{truthful} mechanisms, K and P respectively, that achieve an approximation ratio of $\frac{n+1}{2}$ and $n$, for social cost minimization. In addition, no truthful mechanism can achieve an approximation ratio better than $\frac{n+1}{2}$. Hence, mechanism K is optimal.  
While approximation ratio provides a strong worst-case guarantee, it also limits us to a comprehensive understanding of mechanism performance on various inputs. This paper investigates these two scheduling mechanisms beyond the worst case. We first show that mechanism K achieves a smaller social cost than mechanism P on every input. That is, mechanism K is pointwise better than mechanism P. Next, for each task $j$, when machines' execution costs $t_i^j$ are independent and identically drawn from a task-specific distribution $F^j(t)$, we show that the average-case approximation ratio of mechanism K converges to a constant. This bound is tight for mechanism K. For a better understanding of this distribution dependent constant, on the one hand, we estimate its value by plugging in a few common distributions; on the other, we show that this converging bound improves a known bound \cite{DBLP:conf/aaai/Zhang18} which only captures the single-task setting. Last, we find that the average-case approximation ratio of mechanism P converges to the same constant.

\end{abstract}

\section{Introduction}
The aim of Algorithmic Mechanism Design \cite{NR99,AGT,PT:09} is to design a system for multiple self-interested participants, such that the participants' self-interested actions in equilibrium lead to a good system performance. More specifically, a designer devises a \emph{mechanism} that collects agents' reports as input and computes an output. A mechanism is \emph{truthful} when each agent gets the highest utility by reporting truthfully, given whatever inputs of the other agents. Truthfulness has become a hard constraint in mechanism design. On the one hand, it simplifies agents' decision-making process; on the other hand, it allows the designer to predict the outcome and evaluate mechanism performance. A truthful mechanism is evaluated by the extent to which it achieves its goals, such as minimizing social cost or maximizing social welfare. For this purpose, we use \textit{approximation ratio} to compare the performance of a truthful mechanism against the optimal solution, over all possible inputs.  

Even though the volume of worst-case inputs may only be a drop in the ocean compared to the whole input space, these inputs determine a mechanism's approximation ratio. In other words, the approximation ratio is a worst-case measure. It strongly guarantees the performance of a mechanism no matter what input is provided. However, it is recognized that worst-case instances do not necessarily represent real-world cases. Hence, a worst-case measure is not sufficient for a comprehensive understanding of mechanism performance on the whole input space. Each and every time we adopt a mechanism, given the input uncertainties, we are concerned about its output versus the optimal solution in expectation. Therefore, the \textit{average-case approximation ratio} is an established complement (\cite{DBLP:conf/mfcs/DengG017,DBLP:conf/aaai/Zhang18}) to the worst-case ratio. 

In this paper, we examine the two mechanisms that are devised by \citeauthor{DBLP:journals/mst/Koutsoupias14} [2014] for the problem of scheduling unrelated machines without money. 
The problem is to schedule a set of tasks on $n$ unrelated machines, in order to minimize the total time. The time needed by a machine to execute a task is the machine's private information. The machines are self-interested and want to minimize their execution time. They may achieve this by misreporting their processing time to the mechanism. However, the machines are bound by their declarations. More specifically, in the case when a machine declares a longer time than its actual time for a task, and it is allocated the task, then in practice, its processing time must be the declared value. This is in the spirit that machines have been observed during the execution of the task and cannot afford to be caught lying about the execution times (an unaffordable penalty would apply). \citeauthor{DBLP:journals/mst/Koutsoupias14} [2014] devises a truthful-in-expectation mechanism, K, which achieves the tight approximation ratio bound of $\frac{n+1}{2}$. It beats the other mechanism, P, that allocates the tasks to machines with probability inversely proportional to the declared value. Later, \citeauthor{DBLP:conf/aaai/Zhang18} [2018] shows that the average-case approximation ratio of mechanism K is upper bound by a constant in the single-task setting. 

\subsection{Our contribution}

Given the fairly unsatisfying worst-case guarantees from \cite{DBLP:journals/mst/Koutsoupias14}, it is desirable to consider other models for the input.
We investigate scheduling mechanisms K and P on a pointwise comparison and through an average-case analysis lens. The contribution of this paper is twofold:
\begin{itemize}
\item On a pointwise comparison, mechanism K is better than mechanism P. We show this by proving that the social cost attainable by mechanism K is always smaller than that of mechanism P, for any input. 
\item Through an average-case analysis lens, we show that the average-case approximation ratio of mechanism K converges to a constant, when machines' execution costs $t_i^j$ are independent and identically drawn from a task-specific distribution $F^j(t)$. This bound is tight for mechanism K. 
Furthermore, 
   \begin{itemize}
      \item For a better understanding of this distribution dependent constant, we estimate its value by plugging in a few standard distributions.
      \item We show that this bound improves a known constant bound in \cite{DBLP:conf/aaai/Zhang18} which only captures the single-task setting.
      \item Surprisingly, the average-case approximation ratio of mechanism P converges to the same constant, and hence both mechanisms perform comparably well for sufficiently large, i.i.d inputs. 
   \end{itemize}
\end{itemize}
In summary, we provide a more in-depth understanding of the performance of these known scheduling mechanisms beyond the worst case.


\subsection{Related work}
Algorithmic Mechanism Design is first studied in \cite{NR99}. The aim of this field is to design algorithms that solve practical problems when the problem inputs are subject to participants' strategic behaviors. It contrasts with classical algorithm design, which takes input data as the truth, even though such data could be manipulated. Therefore, optimizing the solution based on this data does not genuinely solve the underlying problem. 

One typical example of these problems is the scheduling problem. If the machines know how a mechanism allocates tasks to them for execution, to minimize the total execution time a machine may be able to execute fewer tasks by lying about its execution costs. Hence, the first challenge is to design a mechanism in which telling the truth is the dominant strategy of these machines. In addition, on some occasions, the use of payments to incentivize the machines may not be feasible. Thus, approximate mechanism design without money \cite{PT:09} becomes an important area of study. For a more elaborate discussion on algorithmic mechanism design, we refer the readers to \cite{AGT}. 

For the scheduling problem where using payments is not feasible, \citeauthor{DBLP:journals/mst/Koutsoupias14} [2014] first considers the setting in which the machines are bound by their declarations. This is similar to \cite{DBLP:conf/wine/KovacsMV15} where the authors introduce the scenario that a mechanism can check the declaration of the agents at running time and guarantee that those who overreport their costs end up paying the exaggerated costs. These penalties can be enforced whenever costs can be measured and certified. In addition, there are similar notions such as impositions (in \cite{DBLP:journals/tcs/FotakisT13}) for the facility location problems, as well as verification (in \cite{DBLP:journals/jcss/AulettaPPP09}). \citeauthor{DBLP:journals/geb/PennaV14} [2014] present a general construction of collusion-resistant mechanisms with verification that return optimal solutions for a wide class of mechanism design problems, including the scheduling problem. \citeauthor{DBLP:conf/ijcai/KovacsV15} [2015] characterize truthful mechanisms in scheduling problems. \citeauthor{DBLP:conf/aaai/ConitzerV14} [2014] focus on scheduling with uncertain execution time. \citeauthor{RePEc:aea:aecrev:v:107:y:2017:i:11:p:3257-87} [2017] introduces the concept of Obviously Strategy-Proof mechanisms which stems from the observation that the practical evidence of truthfulness depends on the implementation details and agents' rationality. Such designs have been adopted for scheduling problems \cite{DBLP:conf/esa/FerraioliMPV19}.

In \cite{DBLP:journals/mst/Koutsoupias14}, for the single-task setting, the author devises two truthful mechanisms that have an upper bound of $\frac{n+1}{2}$ and $n$, respectively. It is then proved that the lower bound of the problem is $\frac{n+1}{2}$. Hence, the mechanism (denoted by K) that has an approximation ratio bound $\frac{n+1}{2}$ is optimal. The other mechanism (denoted by P) with approximation ratio $n$, which allocates the task to a machine with probability inversely proportional to its declared value, does not attract further attention. By running the optimal mechanism independently on multiple tasks, we get a tight bound of $\frac{n+1}{2}$ for social cost minimization and an upper bound of $\frac{n(n+1)}{2}$ for the makespan minimization. \citeauthor{DBLP:conf/aaai/Zhang18} [2018] investigates the optimal mechanism and shows that its approximation ratio is upper bound by a constant in expectation for when the costs are independent and identically distributed random variables. However, it is unclear whether the optimal mechanism K is better than P beyond worst cases and whether the answer to this question depends on the distributions that the machines' costs follow. 


For the model presented in the seminal paper \cite{NR99} where payments are allowed to facilitate designing truthful mechanisms, the best known upper bound is achieved by allocating each task independently using the classical VCG mechanism, while the best known lower bound is 2.61 \cite{DBLP:journals/algorithmica/KoutsoupiasV13}. \citeauthor{DBLP:journals/mor/AshlagiDL12} [2012] prove that the upper bound of $n$ is tight for anonymous mechanisms. For randomized mechanisms, the best known upper bound is $\frac{n+1}{2}$ shown by \cite{DBLP:conf/soda/MualemS07}. For the special case of related machines, where the private information of each machine is a single value, \citeauthor{DBLP:conf/focs/ArcherT01} [2001] give a randomized 3-approximation  mechanism. \citeauthor{DBLP:journals/geb/LaviS09} [2009] show a constant approximation ratio for the special case that the processing times of each task can take one of two fixed values. \citeauthor{DBLP:journals/tcs/Yu09} [2009] generalizes this result to two-range-values, while together with \cite{DBLP:conf/stacs/LuY08} and \cite{DBLP:conf/wine/Lu09}, show constant  bounds for the case of two machines. 




\section{Preliminaries}

In the problem of scheduling unrelated machines without payment, there are $n$ self-interested machines (also known as agents) and $m$ tasks. Machine $i$'s cost to execute task $j$ is $t_i^j$, $i\in [n], j\in [m]$. Since there is no monetary payment to compensate the machines' costs, in a game-theoretical setting, self-interested machines may report untrue costs $\tilde{t}_i^j$ to a mechanism, instead of their true costs $t_i^j$, to minimize the execution workload that a mechanism allocates to them. A mechanism is an algorithm that takes costs $\tilde{t}_i^j$ as inputs, and according to which we allocate tasks to machines. We are interested in mechanisms that minimize the total execution cost.

Denote vector ${\bf \tilde{t}}_i = (\tilde{t}_i^1, \cdots, \tilde{t}_i^m) = (\tilde{t}_i^j)_{j \in [m]}$ the execution costs that machine $i$ reports to a mechanism and ${\bf t}_i = (t_i^j)_{j\in [m]}$ their true execution costs. Let $\tilde{\bf t} = ({\bf \tilde{t}}_i)_{i\in[n]}$ the reports of all machines' costs. Given the declared values $\tilde{t}_i^j$, denote the output of a mechanism by $\mathbf p=(p_i^j)_{i\in[n]. j\in[m]}$, where $p_i^j$ is the probability of machine $i$ getting allocated to execute task $j$ in randomized mechanisms. We follow the literature \cite{DBLP:journals/mst/Koutsoupias14} in which the cost of machine $i$ for task $j$ is $\max\{t_i^j,\tilde{t}_i^j\}$. So, if a machine $i$ declares $\tilde{t}_i^j \ge t_i^j$ and it is allocated the task, then its actual cost is the declared value $\tilde{t}_i^j$ and not $t_i^j$. This is on the understanding that machines are being observed during the execution of the task and cannot afford to be caught lying about the execution times (a high penalty would apply). Therefore, the expected cost of machine $i$ is $c_i=c_i({\bf t}_i,\tilde{\bf t}) = \sum_{j \in[m]} p_i^j(\tilde{\bf t}) \max(t_i^j,\tilde{t}_i^j)$. In approximate mechanism design, we restrict our interest to the class of \emph{truthful mechanisms}. A mechanism is \textit{truthful} if for any report of other machines, the expected cost $c_i$ of machine $i$ is minimized when ${\bf \tilde{t}}_i = {\bf t}_i$. 

The canonical measure of efficiency of a truthful mechanism $\mathrm{M}$ is the \emph{worst-case approximation ratio},
\begin{equation*}
r_{\text{worst}}(\mathrm{M}) = \sup_{{\bf t} \in \mathcal{T}} \frac{SC_{\mathrm M}({\bf t})}{SC_{\mathrm{OPT}}({\bf t})},
\end{equation*}
where $SC_{\mathrm{OPT}}({\bf t})= \min_{\bf p \in \mathcal{P}}\sum_{i=1}^{n}c_i$ is the optimal social cost; $SC_{\mathrm M}({\bf t})$ is the social cost of the mechanism $\mathrm{M}$ on input $\bf t$; and $\mathcal{T}$ is the input space and $\mathcal{P}$ is the allocation probability space. This ratio compares the social cost of a truthful mechanism $\mathrm{M}$ against the social cost of the optimal cost $\mathrm{OPT}$ over all possible inputs $\bf t$. 

Analogous to the worst-case ratio, the \emph{average-case approximation ratio} is a pointwise metric that measures, in expectation, the ratio of the two costs. That is,
\begin{equation*}
r_{\text{average}}(\mathrm{M}) = \mathop{\mathbb{E}}_{{\bf t} \sim \mathrm{D}} \left[ \frac{SC_{\mathrm{M}}(\mathrm{\mathbf t})}{SC_{\mathrm{OPT}}(\mathrm{\mathbf t})} \right] ,
\end{equation*}
where the input $\mathbf t$ follows a distribution $\mathrm{D}$. 

\citeauthor{DBLP:journals/mst/Koutsoupias14} [2014] devises the following two truthful mechanisms for the case that there is only one task. For simplicity, in the single-task setting where $m=1$, we drop the superscript $j$. So, the inputs of the mechanisms are costs $\mathbf t=(t_1,\ldots,t_n)$. 

{\bf Mechanism K}: 
Given the machines' costs $\mathbf t=(t_1,\ldots,t_n)$, without loss of generality, let the values of $t_i$'s be in ascending order $0< t_1 \le t_2 \le \cdots \le t_n$. Then the allocation probabilities are 
\begin{align*}
p_1^K &=\frac{1}{t_1} \int_{0}^{t_1} \prod_{i=2}^{n} \Big(1-\frac{y}{t_i}\Big) dy, \\
p_k^K &=\frac{1}{t_1 t_k} \int_{0}^{t_1} \int_{0}^{y} \prod_{\substack{i=2,\ldots,n \\ i\neq k}} \Big( 1-\frac{x}{t_i} \Big) dx dy, \text{for} \,\ k\neq 1.
\end{align*}

\citeauthor{DBLP:journals/mst/Koutsoupias14} [2014] proves that this mechanism achieves a worst-case approximation ratio upper bound of $\frac{n+1}{2}$ and no other truthful mechanism can do better than $\frac{n+1}{2}$ . For the multi-task setting, we can obtain a truthful mechanism by running mechanism K independently for every task and it retains the approximation ratio $\frac{n+1}{2}$ for social cost minimization. \citeauthor{DBLP:conf/aaai/Zhang18} [2018] further shows that the average-case approximation ratio of mechanism K is upper bound by a constant, when the costs are independent and identically distributed, for the single-task case. 


{\bf Mechanism P}: 
Given the input $\mathbf t=(t_1,\ldots,t_n)$, the mechanism allocates the task to machine $i$ with probability \textit{inversely proportional} to the costs $t_i$, i.e., $p_i^P = \frac{t_i^{-1}}{\sum_{k=1}^n t_k^{-1}}$, $i \in [n]$. 

It is easy to verify that this mechanism is truthful as well. The worst-case approximation ratio of mechanism P is $n$ which is worse than mechanism K.




\section{Pointwise Comparison}
In this section, we show that the social cost attainable by mechanism K is always smaller than that of mechanism P, for any input. That is, mechanism K is universally better than mechanism P. 

For any single task, given the inputs $\mathbf t=(t_1,\ldots,t_n)$, let $t_{(1)} \leq t_{(2)} \leq \cdots \leq t_{(n)}$ be the ascending order statistics of $t_1,..., t_n$. First, we notice that we can rewrite the allocation probabilities of mechanism K so that the double integrals become single integrals. That is, for $k=2,\cdots,n$, 
\begin{align*}
p_{(k)}^K &=\frac{1}{ t_{(1)} t_{(k)}} \int_{0}^{t_{(1)}} \int_{0}^{y} \prod_{\substack{i=2,\ldots,n \\ i\neq k}} \Big( 1-\frac{x}{t_{(i)}} \Big) dx dy  \\
    &= \frac{1}{t_{(1)} t_{(k)}} \int_{0}^{t_{(1)}} \int_{x}^{t_{(1)}} \prod_{\substack{i=2,\ldots,n \\ i\neq k}} \Big( 1-\frac{x}{t_{(i)}} \Big) dy dx  \\
    &= \frac{1}{ t_{(k)}} \int_{0}^{t_{(1)}}  \prod_{\substack{ i\neq k}} \Big( 1-\frac{x}{t_{(i)}} \Big) dx .
\end{align*}

Therefore, the allocation probabilities of mechanism K can be summarized in one equation of single integrals as, 
\begin{align}\label{def: M_allocation_prob0}
p^K_{(k)} = \frac{1}{ t_{(k)}} \int_{0}^{t_{(1)}}  \prod_{\substack{ i\neq k}} \Big( 1-\frac{x}{t_{(i)}} \Big) dx, \,\ \forall k \in [n].
\end{align}

Recall the allocation probabilities of mechanism P are 
\begin{align*}
p^P_{(k)} = \frac{t_{(k)}^{-1}}{\sum_{i=1}^n t_{(i)}^{-1}}, \,\ \forall k \in [n] .
\end{align*}

With the succinct representation of $p^K_{(k)}$ and expressing the allocation probabilities of both mechanisms in terms of the order statistics of the input, we can show that there exists a threshold index, before which the allocation probabilities of mechanism K are larger than that of mechanism P, and vice versa afterwards.  

\begin{lemma}
\label{lemma: compare_MandP}
Given any input $\mathbf t=(t_{(1)},\ldots,t_{(n)})$, there exists an index $l$, where $1 \leq l \leq n$,  such that $p^K_{(k)} \geq p^P_{(k)}$ for all $k \leq l$ and $p^K_{(k)} < p^P_{(k)}$ for all $k> l$.
\end{lemma}

\begin{proof}
For any $k$ and $k'$ such that $1 \leq k < k' \leq n$, since $ 1 - \frac{x}{t_{(k)}} \le 1 - \frac{x}{t_{(k')}}$, $\forall x>0$, we have 
\begin{align*}
\frac{p^K_{(k)}}{p^K_{(k')}} = \frac{t_{(k)}^{-1}}{t_{(k')}^{-1}} \cdot \frac{\int_{0}^{t_{(1)}}  \prod_{\substack{ i\neq k}} \Big( 1-\frac{x}{t_{(i)}} \Big) dx}{\int_{0}^{t_{(1)}}  \prod_{\substack{ i\neq k'}} \Big( 1-\frac{x}{t_{(i)}} \Big) dx} \ge \frac{t_{(k)}^{-1}}{t_{(k')}^{-1}} \cdot 1 = \frac{p^P_{(k)}}{p^P_{(k')}}.
\end{align*}
So, when $p^K_{(k')} > p^P_{(k')}$, we have 
\begin{equation}
 p^K_{(k)} \ge p^K_{(k')} \cdot \frac{p^P_{(k)}}{p^P_{(k')}} = p^P_{(k)} \cdot  \frac{p^K_{(k')}}{p^P_{(k')}}  > p^P_{(k)}.
 \label{leq: M_P}
\end{equation}
When $p^K_{(k)} < p^P_{(k)}$, we have 
\begin{equation}
  p^K_{(k')} \le p^K_{(k)} \cdot \frac{p^P_{(k')}}{p^P_{(k)}} =  p^P_{(k')} \cdot \frac{p^K_{(k)}}{p^P_{(k)}}  < p^P_{(k')} .
   \label{geq: M_P}
\end{equation}

Since $\sum_{k =1}^n p^K_{(k)} = \sum_{k =1}^n p^P_{(k)} =1$, there exists $l$, where $1 \leq l \leq n$, such that $p^K_{(l)} \geq p^P_{(l)}$. Let $l$ be the largest integer between $1$ and $n$ such that $p^K_{(l)} \geq p^P_{(l)}$. By combining (\ref{leq: M_P}) and (\ref{geq: M_P}), we conclude that $p^K_{(k)} \geq p^P_{(k)}$ for all $k \leq l$ and $p^K_{(k)} < p^P_{(k)}$ for all $k> l$.
\end{proof}

Next, we are ready to show that mechanism K is better than mechanism P, for every input. 
\begin{theorem}
\label{theo: MandP}
Given any input $\mathbf t=(t_{(1)},\ldots,t_{(n)})$, mechanism K is always better than mechanism P. That is,
\begin{align*}
SC_{K}({\bf t}) \le SC_{P}({\bf t}). 
\end{align*}
\end{theorem}
\begin{proof}
Since $\sum_{k =1}^n p^K_{(k)} = \sum_{k =1}^n p^P_{(k)} =1$, together with Lemma~\ref{lemma: compare_MandP}, we have that $\sum_{k \leq l }  ( p^K_{(k)} - p^P_{(k)} ) = \sum_{k > l } ( p^P_{(k)} - p^K_{(k)}  ) > 0$. Hence,
\begin{align*}
SC_{K}({\bf t}) - SC_{P}({\bf t}) &= \sum_{k=1}^n p^K_{(k)} t_{(k)} - \sum_{k=1}^n p^P_{(k)} t_{(k)} \\
&= \sum_{k \leq l }  ( p^K_{(k)} - p^P_{(k)} ) t_{(k)}  - \sum_{k > l } ( p^P_{(k)} - p^K_{(k)}  )t_{(k)}  \\
& \leq \sum_{k \leq l }  ( p^K_{(k)} - p^P_{(k)} ) t_{(l)} - \sum_{k > l } ( p^P_{(k)} - p^K_{(k)}  )t_{(l +1)} \\
& = \sum_{k \leq l }  ( p^K_{(k)} - p^P_{(k)} ) t_{(l)} - \sum_{k \leq l }  ( p^K_{(k)} - p^P_{(k)} ) t_{(l +1)} \\
& = \sum_{k \leq l }  ( p^K_{(k)} - p^P_{(k)} ) ( t_{(l)} - t_{(l +1)} ) \le 0.
\end{align*}
\end{proof}


\par

\section{Average-case Approximation Ratio}
In this section, we show that the average-case approximation ratio of mechanism K converges to a constant when the number of machines $n$ approaches infinity. 

Denote ${\bf t}^j = ( t^j_1,\cdots,t^j_n )$ the $n$ machines' execution costs to process task $j$. For any task $j \in [m]$, machines' execution costs $t_i^j$ are independent and identically distributed. Denote $t_{\min}^j$ the minimum cost for any machines to process task $j$. Denote $F^j(t)$ the cumulative distribution function (CDF) of distribution $D[t_{\min}^j, +\infty)$. Let $t^j_{(1)} \leq t^j_{(2)} \leq \cdots \leq t^j_{(n)}$ be the ascending order statistics of $t^j_1,..., t^j_n$.  Without loss of generality, we assume that $t^j_{(1)} = t^j_1$. These are valid assumptions as mechanism K is independently running for every task. Let $s$ be a constant, where $s > t^j_{\min}$. Conditional on $t^j_{(1)} \equiv s$, the other $n-1$ machines' costs to process task $j$, $t^j_2,...,t^j_n$, are independent and identically distributed, and follow the conditional CDF,
\begin{equation}
G^j_s (t) := \frac{F^j(t) - F^j(s)}{1 - F^j(s)},\ \text{for}\ t \ge s . \quad 
\label{def: condiCDF}
\end{equation}
For simplicity, denote the expected value
\begin{equation}
\mu^j_{s}: = \mathop{\mathbb{E}}_{t \sim G^j_{s} } \Big[\frac{s}{t}\Big] . \quad 
\label{def: mu}
\end{equation}

The average-case approximation ratio of mechanism K can be rewritten as
\begin{equation}
\begin{split}
& r_{average}(K): = \mathop{\mathbb{E}}_{t_i^j \sim F^j}   \left[ \frac{  \sum_{j \in [m]} SC_{K}({\bf t}^j) }{ \sum_{j \in [m]} SC_{OPT}({\bf t}^j) } \right] .
\label{def: r_M}
\end{split}
\end{equation}


First, by a change of variables, we transform the allocation probabilities of mechanism K in (\ref{def: M_allocation_prob0}) to the following expression for ease of analysis:

\begin{equation}\label{def: M_allocation_prob}
  p_{(k)}^K =  \frac{t_{(1)}}{ t_{(k)}} \int_{0}^{1}  \prod_{\substack{ i\neq k}} \Big( 1-\frac{t_{(1)}}{t_{(i)}} \cdot x \Big) dx , \,\ \forall k \in [n] .  
\end{equation}

Second, we use these allocation probabilities to bound the expected social cost of mechanism K from above and below. With our careful design of the conditional CDF $G^j_s (t)$ and the expected value $\mu^j_{s}$, we are able to show that the lower and upper bounds of the expected social cost of mechanism K only differ slightly. This result will facilitate us squeezing the two bounds and eventually prove the convergence of the average-case approximation ratio.

\begin{lemma}\label{lemma: condi_exp_M}
For any task $j$, given that $t^j_{(1)} \equiv s > t^j_{\min}$ is a constant, the expectation of the social cost of mechanism K satisfies that
\begin{equation*}
\frac{s}{\mu^j_{s}}  \left(1 - \frac{1}{n \cdot \mu^j_{s} }  \right)<  \mathop{\mathbb{E}}  \left[ \left. SC_{K}({\bf t}^j)   \right| t^j_{(1)} = s  \right] < \frac{s}{\mu^j_{s}} .
\end{equation*}
\end{lemma}

\begin{proof}
First, for any task $j$, by substituting allocation probabilities (\ref{def: M_allocation_prob}) into the social cost of mechanism K, we have that
\begin{align}
SC_{K}({\bf t}^j) &= \sum_{k =1}^{n} p_{(k)} t_{(k)}^j \nonumber \\ 
& = t_{(1)}^j  \sum_{k =1}^{n} \int_{0}^{1}  \prod_{\substack{ i\neq k}} \Big( 1-\frac{t_{(1)}^j}{t_{(i)}^j} \cdot x \Big) dx . 
\label{bound: SCM_expectation}
\end{align}

Given that $t^j_{(1)} =s$ is a constant, we plug (\ref{bound: SCM_expectation}) into the expected value of social cost of mechanism K and obtain that


\begin{align}
\mathop{\mathbb{E}}_{t_{(i)}^j, i=2,\cdots,n }  \left[ \left. SC_{K}({\bf t}^j)   \right| t^j_{(1)} =s   \right] &= \mathop{\mathbb{E}} \left[t^j_{(1)}  \sum_{k =1}^{n} \int_{0}^{1}  \left.\prod_{\substack{ i\neq k}} \Big( 1-\frac{t^j_{(1)}}{t^j_{(i)}} \cdot x \Big)dx \right| t^j_{(1)} =s \right] \nonumber \\
&= s  \sum_{k =1}^{n} \int_{0}^{1}  \mathop{\mathbb{E}} \left[\left.\prod_{\substack{ i\neq k}} \Big( 1-\frac{s}{t^j_{i}} \cdot x \Big) \right| t^j_{(1)} =s  \right]dx \nonumber \\ 
&= s  \sum_{k =1}^{n} \int_{0}^{1}  \prod_{\substack{ i\neq k}} \mathop{\mathbb{E}} \left[\left. \Big( 1-\frac{s}{t^j_{i}} \cdot x \Big) \right| t^j_{(1)} =s \right]dx .
\label{eq: SCM prior}
\end{align}
By plugging the expected value (\ref{def: mu}) into (\ref{eq: SCM prior}), together with the fact that $t^j_2,...,t^j_n$ are independent variables, we work out the integration as
\begin{align}
\mathop{\mathbb{E}}_{t_{(i)}^j, i=2,\cdots,n }  \left[ \left. SC_{K}({\bf t}^j)   \right| t^j_{(1)} =s   \right] & = s\left\{\int_{0}^{1}(1- \mu^j_{s}x)^{n-1} dx  + (n-1)\int_{0}^{1}(1-x)(1- \mu^j_{s}x)^{n-2} dx \right\}  \quad \quad \quad \quad\   \nonumber \\ 
& = \frac{s}{\mu^j_{s}} \left[1 - \frac{1}{n} \cdot \frac{(1 - \mu^j_{s}) ( 1 - (1 -\mu^j_{s})^n )}{ \mu^j_{s}}  \right] .
\label{eq: SCM}
\end{align}

Since $t^j_i \ge s$, when $i = 2, \cdots, n$, we know that 
$$0 < \mu^j_{s} < 1.$$ \\
Therefore, 
\begin{align*}
0 < \frac{(1 - \mu^j_{s}) ( 1 - (1 -\mu^j_{s})^n )}{ \mu^j_{s}} < \frac{1}{\mu^j_{s}} .
\end{align*}
Hence, 
\begin{equation*}
\frac{s}{\mu^j_{s}}  \left(1 - \frac{1}{n \cdot \mu^j_{s} }  \right)<  \mathop{\mathbb{E}}  \left[ \left. SC_{K}({\bf t}^j)   \right| t^j_{(1)} = s  \right] < \frac{s}{\mu^j_{s}} .
\end{equation*}
\end{proof}

Finally, we obtain the main result of this section. The proof of the following theorem is deferred to the Appendix.

\begin{theorem}
\label{main_theo: M}
As $n$ approaches infinity, the average-case approximation ratio of mechanism K converges to
\begin{equation*}
r_{average}(K) \sim \frac{\sum_{j \in[m]} (\mathop{\mathbb{E}}_{t \sim F^j } \left[ \frac{1}{t} \right] )^{-1}}{\sum_{j \in[m]} t^{j}_{\min} }.
\end{equation*}
\end{theorem}


\subsection{Interpreting the distribution dependent ratio}
In Theorem~\ref{main_theo: M}, we show the convergence of the average-case approximation ratio of mechanism K, for the case that there are multiple tasks. Nevertheless, the ratio depends on the distribution $F_j$ and it is not straightforward to interpret its value. Therefore, we plug in a few standard distributions to better understand this bound in the following. 

\smallskip

\textbf{Pareto Distribution.} 
The Pareto distribution  is a power law distribution that is widely used in the description of social, scientific, geophysical, actuarial, and many other types of observable phenomena. According to \cite{DBLP:conf/sigmetrics/ArlittW96} and \cite{ReedJorgensen} as well as the references therein, the distributions of web server workload and of Internet traffic which uses the TCP protocol match well with the Pareto distribution.  

The cumulative distribution function of a Pareto distribution over the support $[t_{\min}, +\infty)$ is given by
\begin{equation*}
    F(t) = 1 - \Big(  \frac{t_{\min}}{t} \Big)^{\alpha}, \,\ t\ge t_{\min}
\end{equation*}
where $\alpha> 0$ is the tail index of the distribution. Since 
\begin{align*}
    \mathop{\mathbb{E}}_{t \sim F}\left[ \frac{t_{\min}}{t} \right] & = \int_{t_{\min}}^{ +\infty} \frac{ t_{\min} }{t}\ d F(t) \nonumber \\
    & = \alpha \cdot t_{\min}^{\alpha + 1} \int_{t_{\min}}^{ +\infty} \frac{1}{t^{\alpha + 2}} d t \nonumber \\
    & = \frac{ \alpha  }{\alpha + 1} ,
\end{align*}
and the fact that the ratio $r_{average}(K)$ in Theorem \ref{main_theo: M} degenerates to $\left(\mathop{\mathbb{E}}_{t \sim F}\left[ \frac{t_{\min}}{t} \right] \right)^{-1}$ when there is a single task, we have 
\begin{align*}
    r_{average}(K) \sim \left(\mathop{\mathbb{E}}_{t \sim F}\left[ \frac{t_{\min}}{t} \right] \right)^{-1} = \frac{\alpha + 1}{\alpha} . 
\end{align*}
We remark that the tail index $\alpha$ of many Pareto distributions we encounter in practice is a small constant, for which $r_{average}(K)$ will also be a small constant. For example, the Pareto principle, a.k.a. the 80/20 rule, corresponds to the case that $\alpha = \log_4 5 \approx 1.161$, with which $r_{average}(K) \approx 1.861$. 

In case of multiple tasks, the cumulative distribution functions are given by $F^j(t) = 1 - \Big(  \frac{t^j_{\min}}{t} \Big)^{\alpha_j}, \,\ t\ge t^j_{\min}$. The average-case approximation ratio becomes that $r_{average}(K) \sim 1 + \frac{\sum_{j \in[m]} 
t^{j}_{\min} / \alpha_j
}{\sum_{j \in[m]} t^{j}_{\min} }$.

\smallskip

\textbf{Exponential Distribution.} The cumulative distribution function of an exponential distribution over the support $[t_{\min}, +\infty)$ is given by
\begin{equation*}
    F(t) = 1 - e^{ - \lambda ( t - t_{\min})},
\end{equation*}
where $\lambda> 0$ is the  rate parameter of the distribution. We have that 
\begin{align*}
    \mathop{\mathbb{E}}_{t \sim F}\left[ \frac{t_{\min}}{t} \right] & = \int_{t_{\min}}^{ +\infty} \frac{ t_{\min} }{t}\ d F(t) \nonumber \\
    & =  \int_{t_{\min}}^{ +\infty} \frac{ t_{\min} }{t} \lambda \cdot e^{ - \lambda ( t - t_{\min})} dt \nonumber \\
    & \overset{r = \lambda t}{=} \lambda \cdot t_{\min} e^{ \lambda t_{\min} } \int_{\lambda t_{\min} }^{ +\infty} \frac{e^{-r}}{r}  dr \nonumber \\
    &= \lambda \cdot t_{\min} e^{ \lambda t_{\min} } E_1( \lambda t_{\min} ) \,\ ,
\end{align*}
where $E_1(x)= \int_{x }^{ +\infty} \frac{e^{-r}}{r}  dr = \frac{e^{-x}}{x} \sum_{n=0}^{+\infty} \frac{n!}{(-x)^n}$ is the exponential integral function. 
Take $\lambda t_{\min} = 2$, for example, $r_{average}(K) \sim (2e^2E_1(2))^{-1} \approx 1.384$.

To obtain a simpler representation of the bound, we note that \citeauthor{abramowitz1988handbook} [1988] 
bound $E_{1}(x)$ as follows
\begin{equation*}
    \frac{1}{2} e^{-x} \ln \Big( 1 + \frac{2}{x}\Big) < E_1(x) < e^{-x} \ln \Big( 1 + \frac{1}{x}\Big). 
\end{equation*}

Therefore, we have that
\begin{equation*}
    \frac{\lambda t_{\min}}{2} \ln \Big( 1 + \frac{2}{\lambda t_{\min}}\Big) < \mathop{\mathbb{E}}_{t \sim F}\left[ \frac{t_{\min}}{t} \right] <  \lambda t_{\min} \ln \Big( 1 + \frac{1}{\lambda t_{\min}}\Big). 
\end{equation*}

Hence, 
\begin{equation*}
    r_{average}(K) \sim \left(\mathop{\mathbb{E}}_{t \sim F}\left[ \frac{t_{\min}}{t} \right] \right)^{-1} \leq \frac{1}{\frac{\lambda t_{\min}}{2} \ln \Big( 1 + \frac{2}{\lambda t_{\min}}\Big) }.
\end{equation*}


\smallskip


\subsection{An Improved Bound for any Distribution.}
Note that when there is a single task, the ratio $r_{average}(K)$ in Theorem \ref{main_theo: M} degenerates to $\left(\mathop{\mathbb{E}}_{t \sim F}\left[ \frac{t_{\min}}{t} \right] \right)^{-1}$. 
This bound improves the constant bound result in \cite{DBLP:conf/aaai/Zhang18}, for any distribution. For a comparison, recall the following result.

\begin{theorem} \cite{DBLP:conf/aaai/Zhang18}
For any distribution $F$ on $[t_{\min}, +\infty]$ and a constant $h$ such that $F(h t_{\min}) \ge \frac{11}{12}$, it holds that~ $r_{average}(K) < 2h + 1.33$.
\end{theorem}

\medskip

We break down the calculation of the expected value to two intervals $t \leq h t_{\min}$ and $t > h t_{\min}$, and use the first interval to bound the expectation as follow. 
\begin{equation*}
\mathop{\mathbb{E}}_{t \sim F}\left[ \frac{t_{\min}}{t} \right] \geq \Pr(t \leq h t_{\min}) \cdot \frac{t_{\min}}{h t_{\min}} + 0 \geq \frac{11}{12} \cdot \frac{1}{h} .
\end{equation*}

So, regardless of the distribution $F$ that costs $t$ follow,
\begin{equation*}
\left(\mathop{\mathbb{E}}_{t \sim F}\left[ \frac{t_{\min}}{t} \right] \right)^{-1} \leq \frac{12}{11}h < 2h + 1.33 .
\end{equation*}



\subsection{Average-case Ratio of Mechanism P} 
Although mechanism P attains a higher social cost than mechanism K on every input, its allocation probabilities are simpler to compute, making it easier to use. Hence, we may wonder whether mechanism P is significantly worse than mechanism K? Surprisingly, we will see that the average-case approximation ratio of mechanism P converges to the same value as mechanism K, which means that both mechanisms perform comparably well for sufficiently large, i.i.d inputs. 
Due to space limitations, the proofs are deferred to the Appendix. 
\begin{theorem}\label{Pratio}
As $n$ approaches infinity, the average-case approximation ratio of mechanism P converges to the same constant as mechanism K. That is, 
\begin{equation*}
r_{average}(P) \sim \frac{\sum_{j\in[m]} (\mathop{\mathbb{E}}_{t \sim F^j } \left[ \frac{1}{t} \right] )^{-1}}{\sum_{j\in[m]} t_{j}^{\min} } .
\end{equation*}
\end{theorem}



\section{Conclusion and Future Work}
This paper focused on the two known truthful mechanisms proposed in \cite{DBLP:journals/mst/Koutsoupias14}, for the problem of scheduling unrelated machines when payment is not feasible. Although mechanism K is optimal for social cost minimization, its approximation ratio is $\frac{n+1}{2}$ which means that there is a large gap between its output and the optimal solution. Following the recent trend of beyond worst-case analysis, we provided an in-depth understanding of the performance of these known mechanisms by considering other models of input. On the one hand, we showed that mechanism K is better than mechanism P, on every single input; on the other, we proved that the average-case approximation ratio of mechanism K converges to a constant when machines' execution costs are independent and identically drawn from a task-specific distribution. This bound improved the constant bound in \cite{DBLP:conf/aaai/Zhang18}, and we provided a better understanding of this distribution dependent constant by plugging in a few standard distributions. Surprisingly, the average-case approximation ratio of mechanism P converges to the same constant, which means that both mechanisms perform comparably well for sufficiently large, i.i.d inputs. 

We hope that the research beyond worst-case analysis in mechanism design will flourish in the following years. To this end, we list a number of open questions along with this work. First, what is the lower bound of the average-case approximation ratio for this problem? Second, is it possible to design new mechanisms with a better average-case approximation ratio? As we have seen in this paper and average-case analysis in algorithm design, the answer to these two questions is likely to depend on the input distribution as it is rare to see a universal analysis for all distributions. Third, can we derive an average-case approximation ratio bound when we drop the identical distribution assumption? That is, the machines' costs are drawn from an independent but not necessarily identical distribution. More generally, what if the costs are correlated? Fourth, when we shift from social cost minimization to makespan minimization, the same questions remain and are probably more challenging. We leave these challenging questions for future work.

\noindent

\newpage

\bibliographystyle{named} 
\bibliography{refs}

\newpage

\section*{Appendix}

{\bf A: Proof of Theorem \ref{main_theo: M}:}

Theorem \ref{main_theo: M} is implied by the following two lemmas.


\begin{lemma}
\label{lemma: nulti_M_upp}
The average-case approximation ratio of mechanism K is upper bound as follows.
\begin{equation*}
r_{average}(K) \leq \frac{\sum_{j\in[m]} (\mathop{\mathbb{E}}_{t \sim F^j } \left[ \frac{1}{t} \right] )^{-1}}{\sum_{j\in[m]} t^{j}_{\min} }. 
\end{equation*}
\end{lemma}

\begin{proof}
According to the average-case approximation ratio (\ref{def: r_M}), we have that
\begin{align*}
& r_{average}(K) = \mathop{\mathbb{E}}_{t_i^j \sim F^j}   \left[ \frac{  \sum_{j \in [m]} SC_{K}({\bf t}^j) }{ \sum_{j \in [m]} SC_{OPT}({\bf t}^j) } \right] = \mathop{\mathbb{E}}_{t_{(1)}^j} \left[ \mathop{\mathbb{E}}_{t_{(i)}^j, i=2,\cdots,n}  \bigg[ \left.\frac{  \sum_{j \in [m]} SC_{K}({\bf t}^j) }{ \sum_{j \in [m]} SC_{OPT}({\bf t}^j) }  \right| t_{(1)}^j  \bigg] \right] .
\end{align*}
Since $\sum_{j \in [m]} SC_{OPT}({\bf t}^j) = \sum_{j \in [m]} t^j_{(1)}$ and it is independent of $t_{(i)}^j, i=2,\cdots,n$, together with Lemma~\ref{lemma: condi_exp_M}, we can upper bound the ratio by
\begin{equation*}
r_{average}(K)  \leq \mathop{\mathbb{E}}_{t_{(1)}^j, j \in [m]} \left[   \frac{ 1 }{ \sum_{j \in [m]} t^j_{(1)} }  \cdot  \sum_{j \in [m]} \frac{t^j_{(1)}}{\mu^j_{t^j_{(1)}}}   \right]  .
\end{equation*}
For any $\delta > 0$, let $\{ t^j_{(1)} > t^j_{\min} + \delta \}$ denotes the set of events that $t^j_{(1)} > t^j_{\min} + \delta$. On this set, we can bound the expected value by the worst-case ratio $\frac{n+1}{2}$. Similarly, we define the set $\{ t^j_{(1)} \leq t^j_{\min} + \delta \}$. Then, we break down the calculation of the expectation via

\begin{equation*}
\begin{split}
r_{average}(K) &\leq \Pr \Big( \bigcup_{j \in[m]} \{ t^j_{(1)} > t^j_{\min} + \delta \} \Big) \cdot \frac{n+1}{2} + \Pr \Big( \bigcap_{j \in[m]} \{ t^j_{(1)} \leq t^j_{\min} + \delta \} \Big) \cdot \frac{\sum\limits_{j \in [m]} \max\limits_{ t^j_{\min} \leq s \leq t^j_{\min} + \delta } \Big\{ 
\frac{ s }{\mu^j_{s}}   \Big\}}{\sum_{j \in [m]} t^j_{\min} } \\
& \leq \Pr \Big( \bigcup_{j \in[m]} \{ t^j_{(1)} > t^j_{\min} + \delta \} \Big) \cdot \frac{n+1}{2}  + \frac{\sum_{j \in [m]} \max_{ t^j_{\min} \leq s \leq t^j_{\min} + \delta } \Big\{ 
\frac{ s }{\mu^j_{s}}   \Big\}}{\sum_{j \in [m]} t^j_{\min} } .
\end{split}
\end{equation*}

Since $t^j_{(1)}$ is the first order statistics, and by the independence of $t^j_{(1)}, \forall j$, we have that 
\begin{equation}
\begin{split}
\Pr \Big( \bigcup_{j \in[m]} \{ t^j_{(1)} > t^j_{\min} + \delta \} \Big)
&\leq  \sum_{j \in [m]} \Pr (  t^j_{(1)} > t^j_{\min} + \delta )  \\
& =  \sum_{j \in [m]} \left(1 - F^j(t^j_{\min} + \delta)\right)^n .
\label{eq: M_lowerbound}
\end{split}
\end{equation}

Since the number of tasks $m$ is finite, $\Pr \Big( \bigcup_{j \in[m]} \{ t^j_{(1)} > t^j_{\min} + \delta \} \Big) \cdot \frac{n+1}{2} $ converges to $0$, as $n$ goes to infinity. Therefore, when $n$ approaches infinity, we obtain that
\begin{align}
& r_{average}(K) \leq \frac{\sum_{j \in [m]} \max_{ t^j_{\min} \leq s \leq t^j_{\min} + \delta } \big\{ 
\frac{ s }{\mu^j_{s}}   \big\}}{\sum_{j \in [m]} t^j_{\min} }.
\end{align}
Last, we complete the proof by letting $\delta$ approach 0.
\end{proof}

\begin{lemma}
\label{lemma: nulti_M_low}
The average-case approximation ratio of mechanism K is lower bounded as follows.
\begin{equation*}
r_{average}(K) \geq \frac{\sum_{j \in[m]} (\mathop{\mathbb{E}}_{t \sim F^j } \left[ \frac{1}{t} \right] )^{-1}}{\sum_{j \in[m]} t^{j}_{\min} }. 
\end{equation*}
\end{lemma}

\begin{proof}
For any $\delta > 0$,
\begin{equation*}
\begin{split}
 r_{average}(K) = \mathop{\mathbb{E}}_{t_i^j \sim F^j}   \left[ \frac{  \sum_{j \in [m]} SC_{K}({\bf t}^j) }{ \sum_{j \in [m]} SC_{OPT}({\bf t}^j) } \right] &= \mathop{\mathbb{E}}_{t_{(1)}^j} \Bigg[ \mathop{\mathbb{E}}_{t^j_{(i)}, i=2,\cdots,n}  \bigg[ \left.\frac{  \sum_{j \in [m]} SC_{K}({\bf t}^j) }{ \sum_{j \in [m]} SC_{OPT}({\bf t}^j) }  \right| t_{(1)}^j  \bigg] \Bigg] \\
& \geq \mathop{\mathbb{E}}_{t_{(1)}^j} \Bigg[   \frac{ 1 }{ \sum_{j \in [m]} t^j_{(1)} }  \cdot  \sum_{j \in [m]} \frac{t^j_{(1)}}{\mu^j_{t^j_{(1)}}}  \bigg(1 - \frac{1}{n \cdot \mu^j_{t^j_{(1)}}} \bigg)  \Bigg]  \\
& \geq \Pr \Big( \bigcap_{j \in[m]} \{ t^j_{(1)} \leq t^j_{\min} + \delta \} \Big) \cdot \frac{\sum\limits_{j \in [m]} \min\limits_{ t^j_{\min} \leq s \leq t^j_{\min} + \delta } \bigg\{ 
\frac{ s }{\mu^j_{s}}  \Big(1 - \frac{1}{n \cdot \mu^j_{s}} \Big)  \bigg\}}{  m \delta + \sum_{j \in [m]} t^j_{\min}  } ,
\end{split}
\end{equation*}
where the first inequality is due to Lemma~\ref{lemma: condi_exp_M}.

For $t^j_{\min} \leq s \leq t^j_{\min} + \delta $, since $\mu^j_s$ is bounded and $\min_{ t^j_{\min} \leq s \leq t^j_{\min} + \delta } \{ 
\mu^j_s \} > 0 $, we have that $\frac{1}{n  \cdot \mu^j_{s}} $ converges to $0$ for all $j$, when $n$ approaches infinity.

Furthermore, since $t^j_{(1)}$ is the first order statistics, we have that 
\begin{align}\label{eq: orderstatistics}
\Pr \Big( \bigcap_{j \in[m]} \{ t^j_{(1)} \leq t^j_{\min} + \delta \} \Big) & = 1 - \Pr \Big( \bigcup_{j \in[m]} \{ t^j_{(1)} > t^j_{\min} + \delta \} \Big) \nonumber  \\
&\geq 1 - \sum_{j \in [m]} \Pr (  t^j_{(1)} > t^j_{\min} + \delta )  \nonumber  \\
& = 1 - \sum_{j \in [m]} \left(1 - F^j(t^j_{\min} + \delta)\right)^n .
\end{align}

Since the number of tasks $m$ is finite, $\Pr \Big( \bigcap_{j \in[m]} \{ t^j_{(1)} \leq t_{\min} + \delta \} \Big)$ converges to $1$, as $n$ goes to infinity. Therefore, when $n$ approaches infinity, we obtain that
\begin{align*}
& r_{average}(K) \geq \frac{\sum_{j \in [m]} \min_{ t^j_{\min} \leq s \leq t^j_{\min} + \delta } \bigg\{ 
\frac{ s }{\mu^j_{s}}    \bigg\}}{  m \delta + \sum_{j \in [m]} t^j_{\min}  }.
\end{align*}
Last, we complete the proof by letting $\delta$ approaches 0.
\end{proof}

Combining Lemma \ref{lemma: nulti_M_upp} and Lemma \ref{lemma: nulti_M_low}, we obtain Theorem \ref{main_theo: M} as a corollary.

\bigskip

{\bf B: Proof of Theorem \ref{Pratio}:}

First, we show the convergence of the expected social cost of mechanism P, for any task $j$. 

\begin{lemma}\label{lemma: condi_exp_P}
For any task $j$, given that $t^j_{(1)} \equiv s > t^j_{\min}$ is a constant, the social cost of mechanism P, $SC_{P}({\bf t}^j)$, converges to $\frac{s}{\mu^j_{s}}$ almost surely.
\end{lemma}

\begin{proof}
The social cost attainable by mechanism P can be rewritten as follows.
\begin{align*}
SC_{P}( {\bf t}^j) &= \sum_{k =1}^{n} p_{(k)} t_{(k)}^j \nonumber \\ 
& = \sum_{k =1}^{n} \frac{ 1/ t_{(k)}^j }{ \sum_{i =1}^{n} 1/ t^j_{(i)} } \cdot t_{(k)}^j \nonumber \\ 
&  = \frac{t^j_{(1)}}{\frac{1}{n} +  \frac{1}{n} \sum_{i =2}^{n} \frac{t^j_{(1)}}{ t_{(i)}^j} }  \nonumber \\
&  = \frac{s}{\frac{1}{n} +  \frac{1}{n} \sum_{i =2}^{n} \frac{s}{ t_i^j} } .
\end{align*}

Note that $t^j_2,...,t^j_n$ are independent and identically distributed variables and follow the conditional CDF in (\ref{def: condiCDF}). As $n$ approaches infinity, by the law of large numbers, we have that
\begin{align*}
SC_{P}( {\bf t}^j) &= \frac{s}{\frac{1}{n} +  \frac{1}{n} \sum_{i =2}^{n} \frac{s}{ t_i^j} }  \\
&\longrightarrow  \frac{s}{ 0 + \mathbb{E}_{t \sim G^j_{s} } [\frac{s}{t}] } = \frac{s}{\mu^j_{s}}, \ \text{almost surely.}
\end{align*}
\end{proof}

Second, we upper bound the average-case ratio of mechanism P when it is running on multiple tasks. 
\begin{lemma}
\label{lemma: nulti_P_upp}
The average-case approximation ratio of mechanism P is upper bound as follows.
\begin{equation*}
r_{average}(P) \leq \frac{\sum_{j\in[m]} (\mathop{\mathbb{E}}_{t \sim F^j } \left[ \frac{1}{t} \right] )^{-1}}{\sum_{j\in[m]} t^{j}_{\min} }. 
\end{equation*}
\end{lemma}

\begin{proof}
As we have shown above, $SC_{P}( {\bf t}^j) = \frac{1}{ \frac{1}{n} \sum_{k =1}^{n} \frac{1}{ t_{k}^j} }$.
By the law of large numbers, as $n$ approaches infinity, we have
\begin{align*}
SC_{P}( {\bf t}^j) &= \ \frac{1}{ \frac{1}{n} \sum_{k =1}^{n} \frac{1}{ t_{k}^j} }  
\longrightarrow  \frac{1}{ \mathop{\mathbb{E}}_{t \sim F^j}\Big[\frac{1}{t}\Big] }, \ \text{almost surely.}
\end{align*}

Since $ SC_{OPT}({\bf t}^j)  \geq t^{j}_{\min} $, we obtain that

\begin{equation}
\begin{split}
 \,\,\,\,\,\ r_{average}(P) &\leq \frac{\mathop{\mathbb{E}} \left[  \sum_{j\in [m]} SC_{P}({\bf t}^j)  \right]}{\sum_{j\in [m]} t^{j}_{\min}} \\
& \longrightarrow \frac{\sum_{j\in [m]} (\mathop{\mathbb{E}}_{t \sim F^j } \left[ \frac{1}{t} \right] )^{-1}}{\sum_{j \in [m]} t^{j}_{\min} }.
\label{eq: P_lowerbound}
\end{split}
\end{equation}
\end{proof}

Next, we lower bound the average-case ratio of mechanism P.
\begin{lemma}
\label{lemma: nulti_P_low}
The average-case approximation ratio of mechanism $P$ is lower bound as follows.
\begin{equation*}
r_{average}(P) \geq \frac{\sum_{j \in [m]} (\mathop{\mathbb{E}}_{t \sim F^j } \left[ \frac{1}{t} \right] )^{-1}}{\sum_{j \in [m]} t^{j}_{\min} }. 
\end{equation*}
\end{lemma}

\begin{proof}
For any $\delta > 0$, according to Lemma~\ref{lemma: condi_exp_P}, we have 

\begin{align*}
r_{average}(P) &= \mathop{\mathbb{E}}_{t_{(1)}^j} \left[ \mathop{\mathbb{E}}  \bigg[ \left.\frac{  \sum_{j \in [m]} SC_{M}({\bf t}^j) }{ \sum_{j \in [m]} SC_{OPT}({\bf t}^j) }  \right| t_{(1)}^j  \bigg] \right] \\
& \longrightarrow \mathop{\mathbb{E}}_{t_{(1)}^j} \bigg[   \frac{ 1 }{ \sum_{j \in [m]} t^j_{(1)} }  \cdot  \sum_{j \in [m]} \frac{t^j_{(1)}}{\mu^j_{t^j_{(1)}}}   \bigg]  .
\end{align*}

We can lower bound the expectation by computing its value when $t^j_{(1)} \leq t_{\min} + \delta$ only. Therefore,
\begin{equation*}
\begin{split}
r_{average}(P) & \geq \Pr \Big( \bigcap_{j \in[m]} \{ t^j_{(1)} \leq t^j_{\min} + \delta \} \Big) \cdot \frac{ \sum\limits_{j \in [m]} \min\limits_{ t^j_{\min} \leq s \leq t^j_{\min} + \delta } \Big\{ 
\frac{ s }{\mu^j_{s}}   \Big\} }{  m \delta  + \sum_{j \in [m]} t^j_{\min} }.  
\end{split}
\end{equation*}

According to (\ref{eq: orderstatistics}), we know that $\Pr \Big( \bigcap_{j \in[m]} \{ t^j_{(1)} \leq t_{\min} + \delta \} \Big)$ converges to $1$, as $n$ goes to infinity. Therefore, 
\begin{align*}
& r_{average}(P) \geq \frac{ \sum_{j \in [m]} \min_{ t^j_{\min} \leq s \leq t^j_{\min} + \delta } \bigg\{ 
\frac{ s }{\mu^j_{s}}   \bigg\} }{  m \delta  + \sum_{j \in [m]} t^j_{\min} }.
\end{align*}
Last, we complete the proof by letting $\delta$ approaches 0.

\end{proof}

Combining Lemma~\ref{lemma: nulti_P_upp} and Lemma~\ref{lemma: nulti_P_low}, we obtain Theorem \ref{Pratio}.

\end{document}